\pdfoutput=1
\documentclass[journal,10pt]{IEEEtran}  % Comment this line out if you need a4paper %onecolumn,
\IEEEoverridecommandlockouts
\usepackage{cite}
\usepackage{indentfirst}
\usepackage{graphicx}
\usepackage{changepage}
\usepackage{stfloats}
\usepackage{amsfonts,amssymb}
\usepackage{bm}
\usepackage{algorithm}
\usepackage{algorithmic}
\usepackage{amsmath}
\usepackage{amsmath,amssymb,amsfonts}
\usepackage{times}
\usepackage{url}
\usepackage{cite}
\usepackage{textcomp}
\usepackage{xcolor}
\usepackage{color}
\usepackage{setspace}
\usepackage{enumitem}
\usepackage{float}
\usepackage{diagbox}
\usepackage[T1]{fontenc}
\usepackage[utf8]{inputenc}
\usepackage{authblk}
\usepackage{threeparttable}
\usepackage{booktabs}
\usepackage{footnote}
\usepackage{graphicx}
\usepackage{pifont}
\usepackage{subfigure}
\usepackage{mathrsfs}
\newenvironment{proof}{{\indent  \indent \it Proof:}}{\hfill $\blacksquare$}

\begin{document}
	
\title{UAV Trajectory and Beamforming Optimization for Integrated Periodic Sensing and Communication}

\author{
	Kaitao Meng, Qingqing Wu, Shaodan Ma, Wen Chen, and  Tony Q. S. Quek, \textit{Fellow, IEEE}
	\thanks{{K. Meng, Q. Wu, and S. Ma are with the State Key Laboratory of Internet of Things for Smart City, University of Macau, Macau, 999078, China (e-mails: \{kaitaomeng, qingqingwu, shaodanma\}@um.edu.mo). W. Chen is with the Department of Electronic Engineering, Shanghai Jiao Tong University, Shanghai 201210, China (e-mail: wenchen@sjtu.edu.cn). T. Quek (e-mail: tonyquek@sutd.edu.sg) is with Singapore University of Technology and Design, Singapore 487372.}}
	\thanks{This work was supported by the FDCT (under Grant 0119/2020/A3) and the GDST under Grant (2020B1212030003).}
}
%  { \textit{(Corresponding author: Qingqing Wu.)}}

% Shaodan Ma : 0000-0001-5521-3650
\maketitle

%%%%%%%%%%%%%%%%%%%%%%%%%%%%%%%%%%%%%%%%%%%%%%%%%%%%%%%%%%%%%%%%%%%%%%%%%%%%%%%%
\begin{abstract}
Unmanned aerial vehicle (UAV) is expected to bring transformative improvement to the integrated sensing and communication (ISAC) system. However, due to shared spectrum resources, it is challenging to achieve a critical trade-off between these two integrated functionalities. To address this issue, we propose in this paper a new integrated \emph{periodic} sensing and communication mechanism for the UAV-enable ISAC system. Specifically, the user achievable rate is maximized via jointly optimizing UAV trajectory, transmit precoder, and sensing start instant, subject to the sensing frequency and beam pattern gain constraints. Despite that this problem is highly non-convex and involves an infinite number of variables, we obtain the optimal transmit precoder and derive the optimal achievable rate in closed-form for any given UAV location to facilitate the UAV trajectory design. Furthermore, we first prove the structural symmetry between optimal solutions in different ISAC frames without location constraints and then propose a high-quality UAV trajectory and sensing optimization algorithm for the general location-constrained case. Simulation results corroborate the effectiveness of the proposed design and also unveil a more flexible trade-off in ISAC systems over benchmark schemes.

\end{abstract}
%%%%%%%%%%%%%%%%%%%%%%%%%%%%%%%%%%%%%%%%%%%%%%%%%%%%%%%%%%%%%%%%%%%%%%%%%%%%%%%%
\section{Introduction}
\par
Integrated sensing and communication (ISAC) has emerged as a key technology in future wireless networks\cite{Zhang2021Enabling}. However, due to shared spectrum resources and complicated surrounding scatters, it is very challenging to strike a good balance between two potentially conflicting design objectives: high-quality communication service and high-timeliness sensing requirement. Recently, there are some research efforts devoted to unifying sensing and communication to mutually assist each other \cite{Cui2021Integrating}. For example, appropriate radar beam pattern design with the communication requirement guaranteed and Pareto optimization framework of the radar-communication system were presented in \cite{Liu2018MIMO} and \cite{Chen2021Radar}, respectively.
\par 
Driven by on-demand deployment and strong line-of-sight (LoS) links promised by unmanned aerial vehicles (UAVs) \cite{Wu2018Capacity, Meng2021Space}, UAV is expected to be a promising aerial ISAC platform to provide more controllable and balanced integrated service based on the requirements of sensing frequency and communication quality \cite{wei2021safeguarding}. In \cite{lyu2021joint}, a joint beamforming and UAV trajectory optimization algorithm was proposed to improve communication quality while guaranteeing sensing requirement. However, the prior works on ISAC (e.g., \cite{Liu2018MIMO}, \cite{Chen2021Radar}, \cite{lyu2021joint}) mainly focus on improving the performance when both functionalities working simultaneously during the entire considered period, i.e., sensing is always executed along with the communication. This however, may ignore the practical asymmetric sensing and communication requirements. Specifically, the sensing frequency as another important aspect of ISAC systems is practically determined by the timeliness requirement of specific tasks and has not been taken into account in the literature, e.g., a relatively low/high sensing frequency is preferred for low-speed/high-speed target tracking. As such, always forcing both sensing and communication simultaneously may cause excessive sensing, thereby introducing stronger interference and higher energy consumption. Hence, it is crucial to improve sensing efficiency, especially for power limited UAVs. To address this issue, a more general trade-off between sensing and communication needs to be investigated by taking into account the sensing frequency besides the commonly used sensing power, which thus motivates this work.
\begin{figure}[t]
	\centering
	\setlength{\abovecaptionskip}{0.cm}
	\includegraphics[width=8.1cm]{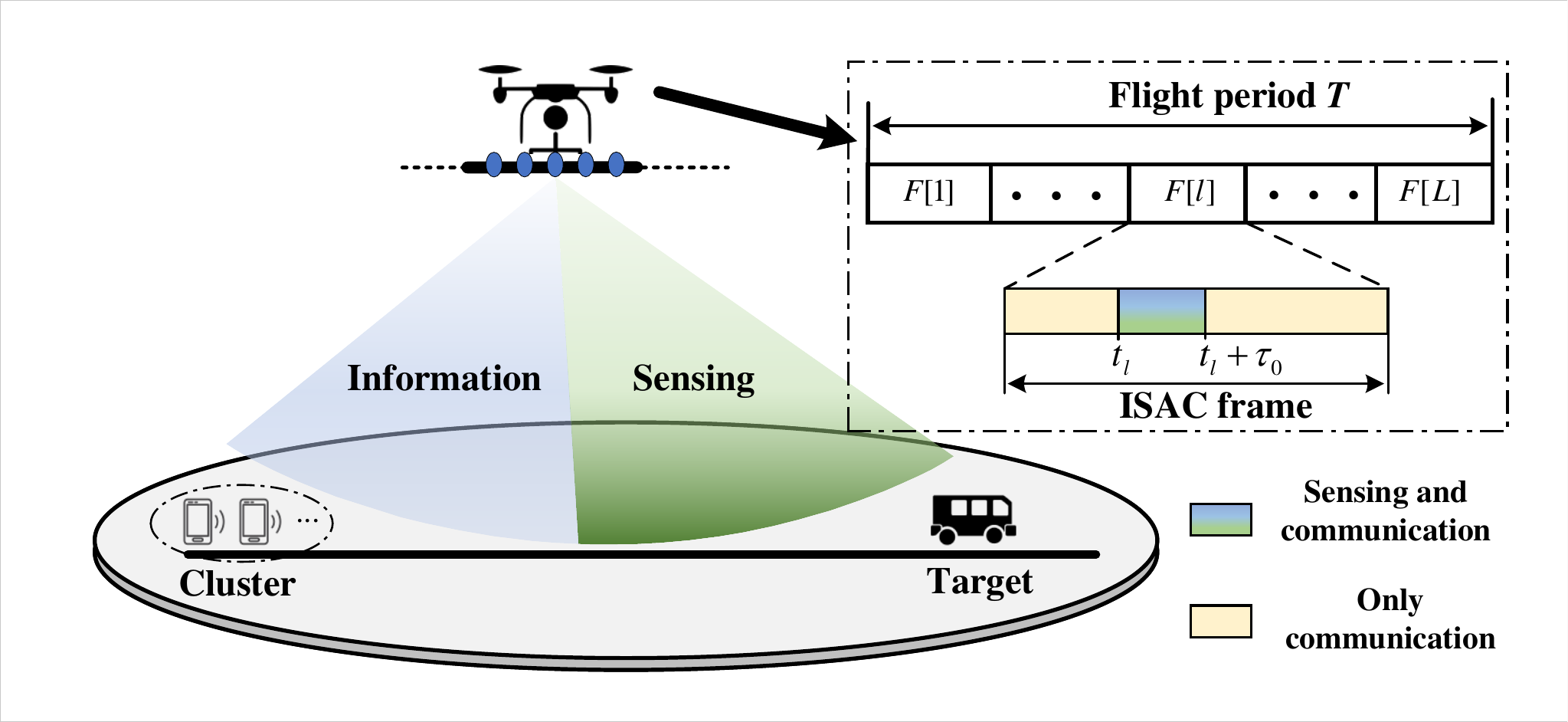}
	\vspace{-2mm}
	\caption{The illustration of integrated periodic sensing and communication.}
	\label{figure1}
	\vspace{-6mm}
\end{figure}
\par 
In this paper, we study a UAV-enabled ISAC system aimed at sensing targets near the ground while providing multi-cast downlink communication for single-antenna users. In particular, to reflect the practically different sensing frequency and sensing power requirements, we propose a periodic ISAC framework where the sensing is periodically executed along with communication. As an initial study, we consider a system model with one target and one cluster of users only, to draw the fundamental insights of UAV-enabled ISAC systems for the considered periodic sensing scenarios, as shown in Fig.~{\ref{figure1}} (e.g., pipeline inspection or traffic monitoring). In this work, the UAV trajectory, transmit beamforming, and sensing instant are jointly optimized to maximize the communication performance while guaranteeing the sensing requirement. Unlike traditional ISAC considered in \cite{lyu2021joint}, which forces data transmitting and radar sensing simultaneously all the time, the proposed scheme offers another opportunity to balance between practical sensing and communication over time. As a result, both standalone communication and always-sensing are special cases of our considered periodic sensing and communication scenarios. The main contribution in this paper is summarized as follows:
\begin{itemize}
	\item To cater for different sensing requirements and avoid excessive energy consumption in practice, we propose an integrated \emph{periodic} sensing and communication (IPSAC) mechanism to provide a more flexible trade-off between sensing and communication over time. We first prove the structural symmetry between optimal solutions in different ISAC frames without location constraints, thereby reducing the algorithm complexity.
	\item We solve the optimal transmit precoder according to the rank-one characteristic and derive the optimal achievable rate in closed-form. Accordingly, an optimal UAV trajectory without location constraints is presented, and a high-quality location-constrained solution is also provided.
\end{itemize}

%%%%%%%%%%%%%%%%%%%%%%%%%%%%%%%%%%%%%%%%%%%%%%%%%%%%%%%%%%%%%%%%%%%%%%%%%%%%%%%
\section{System Model and Problem Formulation}
\vspace{-1mm}
\label{SYSTEM}

As shown in Fig.~\ref{figure1}, we consider a UAV-enabled ISAC system where a UAV equipped with a uniform linear array (ULA) of $M$ antennas is deployed to sense one target near the ground while providing communication service for a cluster of single-antenna users (viewed as one user for ease of analysis, or viewed as an access point for sensory data upload). For convenience, a two-dimensional Cartesian coordinate system is adopted where the linear distance of user and target is referred to as the $x$-axis with two points ${\bm{u}} = [0,0]^T$ and ${\bm{v}} = [{D},0]^T$, respectively. The value of ${\bm{v}}$ is determined based on the specific sensing tasks, which can be set as a sampled position in the region of interest for target detection, or set as a roughly priori-position for target tracking. The UAV is assumed to fly at a constant altitude of $H$ m with a flight duration of $T$ s, and the UAV location is denoted by $x(t)$, where $t \in {\cal{T}} \buildrel \Delta \over = [0,T]$. 

%\subsection{ISAC Frame}

In our proposed IPSAC mechanism, as shown in Fig.~{\ref{figure1}}, the sensing task should be performed at least once in each ISAC frame, and its length is denoted by $T_f$, which is set according to the requirement of task execution frequency. The total frame number $L$ is set as an integer for ease of analysis ($L \cdot T_f=T$), and the index of ISAC frame is denoted by $l \in {\cal{L}}= \{1,\cdots,L\}$. The sensing time in the $l$th ISAC frame is given by ${\cal{T}}_l \buildrel \Delta \over = [t_l, t_l + \tau_0]$, where $\tau_0$ is the sensing period. The minimum sensing period is practically determined by radar bandwidth, waveform, and so on \cite{Cui2021Integrating}. It is assumed that the UAV is hovering during detecting or tracking the target, since the sensing period is generally small and the fixed UAV location can avoid the introduction of complex Doppler shift due to the movement of the UAV.

%\subsection{Communication and Sensing Model}
\par
The communication links between the UAV and the user are assumed to be dominated by the line-of-sight (LoS) component \cite{Wu2018Capacity}. Hence, it is assumed that the aerial-ground channel follows the free-space path loss model, and the channel power gain from the UAV to the user can be given by
	\vspace{-1mm}
\begin{equation}
	\beta_c(x(t))=\beta_{0} d(x(t))^{- 2}=\frac{{{\beta _0}}}{{{H^2} + x{(t)}^2}},
	\vspace{-1mm}
\end{equation}
where $\beta_0$ represents the channel power at the reference distance 1 m. The transmit array response vector of the UAV toward the user's location ${\bm{u}}$ is expressed as
\begin{equation}
	{{\bm{a}}^H}\left(x(t), \bm{u}\right) = [1, \cdots , e^{-j 2 \pi \frac{d}{\lambda} (M-1) \sin (\theta_1 (x(t), \bm{u}))}],
\end{equation}
where $d$ is the half-wavelength antenna spacing, $\lambda$ is the carrier wavelength, $\theta_1$ denotes the elevation angle of the geographical path connecting the UAV to the user, and $\sin (\theta_1 (x(t), \bm{u})) = H / \sqrt{x(t)^2+H^2}$. Therefore, the baseband equivalent channel from the UAV to the user can be given by 
	\vspace{-1mm}
\begin{equation}
	{\bm{h}}^H_c(x(t))=\sqrt{\beta_{{c}}(x(t))} e^{-j \frac{2 \pi d{(x(t))}}{\lambda}} {\bm{a}}^H\left(x(t), \bm{u}\right).
	\vspace{-1mm}
\end{equation}

For practical implementation, the linear precoding is adopted at the UAV-enabled ISAC system. The transmitted signal from the UAV is given by ${\bm{z}}(t) = {{\bm{w}}_c}(t) {s_c}(t)$, where ${s}_c$ and ${\bm{w}}_c(t) \in \mathbb{C}^{M \times 1}$ are the information-bearing signal to the user and its corresponding transmit precoder, ${s}_c \sim \mathcal{C} \mathcal{N}(0,1)$, and $\mathbb{E}(|s_c|^{2}) = 1$. Then, the received signal at user is
\begin{equation}\label{UserReceivedSignal}
	{y}(t) =  {\bm{h}}^H_c(x(t))  {\bm{z}}(t) + {n_c}(t),\forall t \in {\cal{T}},
\end{equation}
where $n_c(t) \sim {\cal{CN}}(0, \sigma^2)$ denotes the additive white Gaussian noise (AWGN) at the user's receiver. Accordingly, the signal-to-noise-ratio (SNR) of the user is given by ${\gamma}(t) = | {{\bm{h}}_c(x(t)) {{\bm{w}}_c}(t)} |^2 / \sigma^2, \forall t \in {\cal{T}}$. As a result, the achievable rate of the user at time $t$ is ${R}(t) = {\log_2 (1 + {\gamma }(t))}$. In our designed ISAC system, the reflected communication signals from the objects of interest can also be analyzed for radar sensing \cite{Cui2021Integrating}, and thus, the communication signals $s_c(t)$ is also exploited for sensing performance. The transmit beam pattern gain from the UAV to the target's location ${\bm{v}}$ can be expressed as
\begin{equation}\label{BeampatternGain}
	\begin{aligned}
		& \Gamma \left(x(t), \bm{v}\right) = E\left[ {{{\left| {\bm{a}}^H(x(t), \bm{v})  {\bm{z}}(t) \right|}^2}} \right] \\
		= & {\bm{a}}^H(x(t), \bm{v})  {   \left( {\bm{w}}_c(t){\bm{w}}^H_c(t)  \right) }   {\bm{a}}(x(t), \bm{v}).
	\end{aligned}
\end{equation}

Our objective is to maximize the achievable rate, subject to the beam pattern gain constraint and the maximum transmit power constraint. The corresponding optimization problem can be formulated as 
\begin{alignat}{2}
\label{P1}
\vspace{-1mm}
(\rm{P1}): \quad & \begin{array}{*{20}{c}}
\mathop {\max }\limits_{ \{{\bm{w}}_c(t)\}, \{x(t)\}, \{t_l\}} \mathop \frac{1}{T} \int_0^T R(t) {\rm{d}}t,
\end{array} & \\ 
\mbox{s.t.}\quad
& \Gamma \left(x(t_l), \bm{v}\right) \ge d(x(t_l), {\bm{v}})^2 {\tilde{\Gamma}}, \forall t_l \in {\cal{T}}_l, l \in {\cal{L}}, & \tag{\ref{P1}a}\\
& \left\|{{\bm{w}}_c}(t)\right\|^2 \le P_{\max}, \forall t \in {\cal{T}}, & \tag{\ref{P1}b}\\
& \left| \dot{x}(t) \right| \le V_{\max}, \forall t \in {\cal{T}}, & \tag{\ref{P1}c}\\
& {{x}}(0) = {{x}}_I,  {{x}}({T}) = {{x}}_F. & \tag{\ref{P1}d}
\vspace{-1mm}
\end{alignat} 
%%%%%%%%%%%%%%%%%%%%%%%%%%%%%%%%%%%%%%%%%%%%%%%%%%%%%%%%%%%%%%%%%%%%%%%%%%%%%%%%
In (P1), (\ref{P1}a) represents the beam pattern gain should be no less than the threshold ${\tilde{\Gamma}}$ under the given sensing frequency and the total power constraint is given in ({\ref{P1}b}). The UAV speed constraint, the initial and final location constraints are given in (\ref{P1}c) and (\ref{P1}d), respectively. Problem (P1) is difficult to solve optimally, since it involves an infinite number of variables $\{x(t)\}$, and the constraints in ({\ref{P1}a}) and objective function are non-convex due to the non-concavity of $\Gamma\left(x(t_l), \bm{v}\right)$ and $R(t)$.

\section{Proposed Solution to (P1)}
\label{ProposedAlgorithm}
In this section, we obtain the optimal transmit precoder by SDR technique together with eigenvalue decomposition, and provide a closed-form user SNR for trajectory optimizing.
\newtheorem{thm}{\bf Lemma}
\newtheorem{remark}{\bf Remark}
\newtheorem{Pro}{\bf Proposition} 
\vspace{-2mm}
\subsection{Optimal Transmit Precoder to (P1)}

\par	
Intuitively, for $t \notin [t_l,t_l+\tau_0]$, $l \in {\cal{L}}$, it is known that the maximum ratio transmission (MRT) is the optimal transmit precoder, i.e., ${\bm{w}}^*_c(t) = {\frac{\sqrt{P_{\max}}{\bm{h}}_c}{\|{\bm{h}}_c\|}}$. In the following, we will focus on analyzing the optimal transmit precoder ${\bm{w}}_c(t)$ for $t \in [t_l,t_l+\tau_0]$, $l \in {\cal{L}}$. Since the maximization of $R(t)$ is equivalent to maximizing the corresponding received signal strength, the $\log$ function is dropped in the objective function for simplicity. Then, for any given UAV location $x(t)$, problem (P1) is reduced to (by dropping the time index $(t)$)
\vspace{-1mm}
\begin{alignat}{2}
\label{P2}
(\rm{P2}): \quad & \begin{array}{*{20}{c}}
	\mathop {\max }\limits_{ {\bm{w}}_c}  \mathop  {{\bm{w}}_c^H} {\bm{h}}_c {\bm{h}}_c^H{\bm{w}}_c,
\end{array} & \\ 
\mbox{s.t.}\quad
& {{\bm{w}}_c^H} {\bm{h}}_r {\bm{h}}_r^H{\bm{w}}_c \ge {\tilde{\Gamma}},  \left\|{{\bm{w}}_c}\right\|^2 \le P_{\max}, & \tag{\ref{P2}a}
\end{alignat} 
where ${\bm{h}}^H_c = {\bm{h}}^H_c(x,{\bm{u}})$, ${\bm{h}}^H_r =\frac{{\bm{a}}^H(x,{\bm{v}})}{d(x, {\bm{v}})}$. Notice that problem (P2) is non-convex and thus is difficult to be optimally solved in general. Hence, we propose to utilize semidefinite relaxation (SDR) technique to solve problem (P2), and then, obtain the optimal transmit precoders based on the rank-one characteristic of optimal solution. Towards this end, we introduce new auxiliary variables ${\bm{W}} = {\bm{w}}{\bm{w}}^H$, where ${\bm{W}} \succeq 0$ and $\operatorname{rank}\left(\bm{W}\right) = 1$. Thus, problem (P2) is equivalent to
\vspace{-1mm}
\begin{alignat}{2}
	\label{P2.1}
	(\rm{P2.1}): \quad & \begin{array}{*{20}{c}}
		\mathop {\max }\limits_{ {\bm{W}}} \mathop {\rm{tr}}\left( {{{\bm{H}}_c}{\bm{W}}} \right),
	\end{array} & \\ 
	\mbox{s.t.}\quad
	& {\rm{tr}}\left( {{{\bm{H}}_r}{\bm{W}}} \right) \ge {\tilde{\Gamma}}, {\rm{tr}}\left( {\bm{W}} \right) \le P_{\max}, & \tag{\ref{P2.1}a} \\
	& {\rm{rank}}\left( {\bm{W}} \right) = 1, & \tag{\ref{P2.1}b} 
	\vspace{-1mm}
\end{alignat} 
where ${{\bm{H}}_c} = {{\bm{h}}_c}{{\bm{h}}^H_c}$ and ${{\bm{H}}_r} = {{\bm{h}}_r}{{\bm{h}}^H_r}$. By ignoring the above rank constraint on ${\bm{W}}$ similarly as in \cite{Xu2014Multiuser}, the semidefinite relaxation (SDR) of problem (P2.1), denoted by (SDR2.1), can be optimally solved via convex optimization solvers. 
\begin{remark}\label{Rank1Illustration}
	For arbitrary user channel and target location, according to Theorem 3.2 in \cite{Huang2010Rank}, problem (SDR2.1) always has an optimal solution $\tilde{\bm{W}}$ satisfying $\operatorname{rank}^2(\tilde {\bm{W}}) \le 2$. Hence, there always exists an optimal solution $\tilde{\bm{W}}$ satisfying $\operatorname{rank}(\tilde{\bm{W}}) = 1$.
\end{remark}

Remark \ref{Rank1Illustration} shows that the optimal transmit precoder ${\bm{w}}_c$ can be recovered by performing eigenvalue decomposition over the obtained rank-one ${\bm{W}}$. However, it is still challenging to obtain the optimal trajectory due to lack of closed-form transmit precoder. To tackle this problem, we provide a closed-form user SNR for trajectory analysis, even if there is no closed-form transmit precoder.
\begin{Pro}\label{OptimalBearfoming}
	For any given UAV location, during sensing time, the optimal user SNR can be given by 
	\vspace{-2mm}
	\begin{equation}\label{OptimalGamma}
		\gamma^* = \left\{ {\begin{array}{*{20}{c}}
				{\frac{\beta_0{P_{\max}}M}{{{x^2} + {H^2}}},}&{\frac{{{M{P_{\max }}\rho^2}}}{{{( {D - x} )}^2} + {H^2}} \ge  {\tilde{\Gamma}}}\\
				{g(x),}&{\rm{Otherwise}}
		\end{array}} \right.,
	\vspace{-2mm}
	\end{equation}
	where $g(x) =  \frac{\gamma_0{({{( {D - x} )}^2} + {H^2}){\left( {\rho  {\sqrt {{{\tilde{\Gamma}}}}} {\rm{ + }}\sqrt { {1 - {\rho ^2}} } {\sqrt {\frac{{M{P_{\max }}}}{{{{( {D - x} )}^2} + {H^2}}}- {{\tilde{\Gamma}}}}}   } \right)^2}}}{{{{x^2} + {H^2}}}}$, $\gamma_0 = \frac{{\beta _0}}{\sigma ^2}$, and $\rho  = \frac{| {\bm{a}}^H(x, {\bm{u}} ) {\bm{a}}(x, {\bm{v}} ) |}{\| {\bm{a}}^H(x, {\bm{u}} ) \|\| {\bm{a}}(x, {\bm{v}} )\| }$.
\end{Pro}
\begin{proof}
	Please refer to Appendix A.
\end{proof}
\par
Proposition 1 explicitly shows that the achievable rate is determined by the correlation between the user's and the target's transmit response vectors. Accordingly, the optimal achievable rate can be directly obtained for any given UAV location $x(t)$, thereby dramatically reducing the complexity of trajectory design. The details are given as follows.

\vspace{-2mm}
\subsection{Optimal Solution Without Location Constraints}
\label{OptimalContinuesTrajectory}
\vspace{-1mm}
To draw important insights into the optimal trajectory design, we first study a special case of (P1) where the initial and final location constraints are ignored, denoted by (SP1). By plugging the optimal SNR in (\ref{OptimalGamma}) into the user achievable rate $R(t)$, (SP1) can be simplified as
\vspace{-1mm}
\begin{alignat}{2}
	\label{SP2}
	(\rm{SP2}): \quad \quad & \begin{array}{*{20}{c}}
		\mathop {\max }\limits_{\{x(t)\}, \{t_l\}} \quad \mathop  \frac{1}{T} \int_{t=0}^{T} \tilde{R}(t) {\rm{d}}t,
	\end{array} & \mbox{s.t.}\quad (\ref{P1}c),  \nonumber
	\vspace{-2mm}
\end{alignat} 
where $\tilde{R}(t)$ represents the achievable rate with optimal transmit precoder ${\bm{w}}_c^*$.
However, it is still non-trivial to solve the optimal UAV trajectory because there is no closed-form integral of the achievable rate. To address this challenge, the following conclusions are given to significantly reduce the complexity of finding the optimal UAV trajectory.
\par
\begin{remark}\label{UnidirectionalTrajectory}
	According to Lemma 3 in \cite{Wu2018Capacity}, there always exists an optimal UAV trajectory ${x^*(t)}$ for problem (SP2) that is unidirectional in the $l$th ISAC frame, i.e., $x^*(t_1) \le x^*(t_2)$, if $t_1 < t_2$, or $x^*(t_1) \ge x^*(t_2)$, if $t_1 < t_2$, $\forall t_1, t_2 \in {\cal{T}}_l$. Hence, we only need to consider the unidirectional UAV trajectory between $[0, D]$ for problem (SP2).
\end{remark}

\begin{thm}\label{EqualForEachFrame}
	There always exists an optimal UAV trajectory satisfying $x(t_1) = x(t_2)$ if $t_1 + t_2 = n T_f$, where $n$ is even.
\end{thm}
\begin{proof}
	Without loss generality, we assume that the total transmission rate of the $l$th ISAC frame is the largest, and its corresponding optimal UAV trajectory is denoted by $\{x^*(t)\}_{t = (l-1)  T_f}^{l  T_f}$. We can always construct an optimal UAV trajectory of the $l+1$th ISAC frame by reversing $x^*(t)$ from the timeline, i.e., $\{x^*(t)\}_{t = l T_f}^{(l-1)  T_f}$, and its corresponding sensing time and transmit precoder at the same location are set to be the same as that of the $l$th ISAC frame. Obviously, the total transmission rate of the constructed solution is no less than that of the original solution, and the constraints in (\ref{P1}c) are also satisfied. Thus, the proof is completed.
\end{proof} 

According to Lemma {\ref{EqualForEachFrame}}, the potential symmetric structure characteristics of the optimal UAV trajectories among adjacent frames are proved. For example, the UAV location at time $t_1 = \alpha T_f$ during the first ISAC frame equals to that at time $t_2 = T_f + (1-\alpha)T_f$ during the second ISAC frame, i.e. $t_1$ and $t_2$ are symmetric with respect to $t = T_f$. Therefore, we can fist optimally solve the UAV trajectory in the first ISAC frame, and then, the optimal UAV trajectory in the second ISAC frame can be constructed by reversing the sequence of that within the first frame. Then, the optimal UAV trajectory of other frames can be obtained similarly.

Based on the above discussion, we can construct an optimal UAV trajectory flying from $x_r$ to $x'_r$ for the first ISAC frame, where $x_r \ge x'_r$. Specifically, the UAV first hovers for a duration of $\tau_0$ s at the sensing location $x_r$ and then flies toward user's direction at its maximum speed for $t \in (\tau_0, T_f]$. Then, the sum-rate based on this trajectory is 
\vspace{-1mm}
\begin{equation}\label{TotalCommunicationRate}
	\begin{aligned}
		{C} = & {\tau_0}g\left( {{x_r}} \right) + \frac{1}{V_{\max}}\int_{{x_r}}^{x'_r} {f\left( x \right)} {\rm{d}}x  + \frac{x'_r}{V_{\max}} f(0),
	\end{aligned}
\vspace{-1mm}
\end{equation}
where ${x'_r} = \max ({x_r} - ({T_f} - {\tau_0}){V_{\max }},0)$, $g\left( x \right)$ represents the achievable rate during sensing at location $x$ (c.f. (\ref{OptimalGamma})), and $f\left( x \right) = {\log _2}\left( {1 + \frac{{\gamma _0}}{{{{x^2} + {H^2}}}}} \right)$ is the achievable rate when only transmitting data to user.
%\begin{figure}[htbp]
%	\centering
%	\setlength{\abovecaptionskip}{0.cm}
%	\includegraphics[width=8cm]{figure3.pdf}
%	\caption{The optimal UAV trajectory structure.}
%	\label{figure3}
%\end{figure}
\begin{thm}\label{FlyingAroundorHovering} % 此处可以改变为多段飞行轨迹的表达式
	There always exists an optimal trajectory in the first ISAC frame satisfying the following structure to problem (SP2), i.e., 
	\vspace{-1.5mm}
	\begin{equation}\label{OptimalSubTrajectory}
		x(t) = \left\{ {\begin{array}{*{20}{c}}
				{  {{x^*_r}} ,}&{t \in [0,\tau_0]}\\
				{\max ( {x^*_r - V_{\max}(t-\tau_0),0} ) ,}&{t \in (\tau_0,T_f]}
		\end{array}} \right.,
	\vspace{-1mm}
	\end{equation}
	where ${x^*_r}$ satisfies the following condition:
	\vspace{-1mm}
	\begin{equation}\label{OptimalSensingCondition}
		{\tau_0}V_{\max}g'\left( {{x^*_r}} \right) = f\left( {\max ( {x^*_r - V_{\max}(T_f-\tau_0),0} )} \right) - f\left( {{x^*_r}} \right).
		\vspace{-1mm}
	\end{equation}
\end{thm}
\begin{proof}
	For any given two UAV trajectories with different sensing locations $x_1$ and $x_2$, the sum-rate difference between these two trajectories can be written as
	\vspace{-1.5mm}
	\begin{equation}
		\begin{aligned}
			\Delta{C} &= {\tau_0}g( {{x_2}} ) - {\tau_0}g( {{x_1}} )  +  \frac{1}{V_{\max}}\int_{{x_1}}^{x_2} {f( x )} {\rm{d}}x \\
			&-\frac{1}{{{V_{\max }}}}\int_{{x'_1}}^{{x'_2}} {f(x)} {\rm{d}}x + \left( {\frac{{{x'_2}}}{{{V_{\max }}}} - \frac{{{x'_1}}}{{{V_{\max }}}}} \right)f(0),
		\end{aligned}
	\vspace{-1.5mm}
	\end{equation}
	where ${x'_2} = \max ({x_2} - ({T_f} - {\tau_0}){V_{\max }},0)$ and ${x'_1} = \max ({x_1} - ({T_f} - {\tau_0}){V_{\max }},0)$. Define $\Delta x = x_2 - x_1$, if the optimal sensing location is $x_1$ or $x_2$, it follows that
	\vspace{-1.5mm}
	\begin{equation}
		\mathop {\lim }\limits_{\Delta x  \to 0 } \frac{ \Delta{C}}{\Delta x} =  {\tau_0}g'( {{x_1}} ) - \frac{1}{V_{\max }}( {f( x'_1 ) - f( {{x_1}} )} ) \\
		= 0.
		\vspace{-1.5mm}
	\end{equation}
	Then, the optimal sensing location should satisfy the condition: ${\tau_0}V_{\max}g'\left( {{x_r}} \right) = f\left( {\max \left( {x'^*_r ,0} \right)} \right) - f\left( {{x^*_r}} \right)$, where $x'^*_r \!=\! {x^*_r} \!-\! {V_{\max }}\left( {{T_l} \!-\! \tau_0} \right)$, and thus, completes the proof.
\end{proof}

Based on Lemma \ref{FlyingAroundorHovering}, the optimal UAV trajectory can be obtained by 1-D search within $[0,D]$ together with checking whether the equation in (\ref{OptimalSensingCondition}) holds. Particularly, if $T_f = \tau_0 $, the optimal trajectory is hovering at the location where $g'\left( {{x_r}} \right) = 0$. As a result, the achievable rate of the optimal UAV trajectory obtained by Lemma \ref{FlyingAroundorHovering} is the upper bound of that for (P1).

\vspace{-1.5mm}
\subsection{Location Constrained Trajectory and Sensing Optimization}
\label{TrajectoryConstruction}
\vspace{-0.5mm}
Mathematically, with initial and final location constraints, the UAV trajectory and sensing start instant optimization sub-problem of (P1) is reduced to
\vspace{-1mm}
\begin{alignat}{2}
	\label{P3}
	(\rm{P3}): \quad & \begin{array}{*{20}{c}}
		\mathop {\max }\limits_{{\{{{x}(t)}\}, \{t_l\}}} \quad \mathop  \frac{1}{T} \int_{t=0}^{T} \tilde{R}(t) {\rm{d}t},
	\end{array} & \mbox{s.t.}\quad (\ref{P1}c)-(\ref{P1}d). \nonumber
\end{alignat} 
\vspace{-0.5mm}
The optimal sensing locations and sensing durations of all ISAC frames are denoted by $\{x^*_{l}\}_{l=1}^L$ and $\{[t^*_{l}, t^*_{l}+\tau_0]\}_{l=1}^L$. We observe that if two consecutive sensing locations on the same side of user, i.e., $x_{l}^*  x_{l+1}^* > 0$, the UAV must fly at its maximum speed from ${x_{l}^*}$ to ${x_{l+1}^*}$. Otherwise, the UAV will hover at the above of user except flying from $x_{l}^*$ to $x_{l+1}^*$ at its maximum speed.

\begin{thm}\label{MultipleSensingLocation}
	The optimal sensing location $x^*_{l}$ and sensing durations $[t^*_{l}, t^*_{l}+\tau_0]$ of the $l$th ISAC frame must satisfy at least one of the following conditions:
	\begin{itemize}
		\item $g'(x^*_{l}) = 0$ (c.f. the definition of $g(x)$ in (\ref{OptimalGamma}));
		\item $t^*_{l} = (l-1)  T_f$ or $t^*_{l} = l  T_f$.
	\end{itemize}
\end{thm}
\begin{proof}
	For any given two sensing locations in the $l$th ISAC frame, the sum-rate difference between these two trajectories is $\tau_0 g(x_{l})$. Hence, if $g'(x^*_{l}) \ne 0$, the sum-rate can be improved by moving the $x^*_{l}$ toward the direction where $g'(x^*_{l}) > 0$. Due to the maximum speed constraints, there may not exist feasible solution satisfying $g'(x^*_{l}) = 0$ within each ISAC frame. In this case, the optimal sensing start instant must be at the start or the end of this ISAC frame, i.e., $t^*_{l} = (l-1)  T_f$ or $t^*_{l} = l  T_f$. Thus, the proof is completed.
\end{proof}

Based on above conclusions, a high-quality solution can be constructed based on the optimal sensing conditions of Lemma {\ref{MultipleSensingLocation}}. Specifically, the UAV will tend to fly at its maximum speed toward location $x_r^*$ in Lemma \ref{FlyingAroundorHovering}, during which, the sensing locations $\{x_{l}\}_{l=1}^{L}$ can be obtained by 1-D search together with checking the conditions in Lemma \ref{MultipleSensingLocation}. After arriving at the location $x_r^*$ , the trajectory is composed of several sub-trajectories with a similar hover-fly-hover structure expressed in (\ref{OptimalSubTrajectory}). Then, the trajectory back to final location can be obtained in a similar way as leaving from initial location. If $| x_I  - x^*_{r}| + | x_F - x^*_r | < L ({T_f - \tau_0})   V_{\max}$, the UAV will fly at its maximum speed from $x_I$ toward location $x_r^*$ and back to $x_F$ at the time ${T \mathord{\left/{\vphantom {T 2}} \right.\kern-\nulldelimiterspace} 2}$.
\vspace{-1mm}
\section{Numerical Results}
\label{Simulations}
\vspace{-1mm}
\par 
In this section, numerical results are provided for characterizing the performance of the proposed IPSAC mechanism. The required parameters are set as follows unless specified otherwise: $P_{\max} = 0.1$ W, $\beta_{0} = -30$ dB, $\sigma^2 = -100$ dB, $\tilde{\Gamma} = 6 \times e ^{-5}$, $T = 500$ s, $T_f = 5$ s,  $\tau_0 = 0.1 $ s, $M = 10$, $d={\lambda  \mathord{\left/{\vphantom {\lambda  2}} \right.\kern-\nulldelimiterspace} 2}$, $V_{\max} = 30$ m/s, $H = 50$ m, $D = 400$ m, and $x_I = x_F = D$. Our proposed schemes for problem (SP2) and (P3) denoted as 1): "Upper bound" and 2): "Proposed" in Fig.~2, respectively, are compared to two benchmarks. 3): {{"Time division"}}: The transmit precoder is set as the MRT to the user and target in a time-division manner; 4): {{"Optimal precoder only"}}: The transmit precoder is obtained by our proposed method. \textbf{Both benchmarks} assume that the UAV performs sensing tasks at the beginning of each ISAC frame, and flies from the initial location to the user during the first-half ISAC frame and then flies to the target during the second half-frame, respectively. 
\par
Figs.~\ref{figure3a}-\ref{figure3b} illustrate the fundamental trade-off among sensing power requirement, sensing frequency, and achievable rate for the considered system. It is observed from Fig.~\ref{figure3a} that the achievable rate gain achieved by our proposed scheme over the "optimal precoder only" scheme increases as the sensing frequency decreases, as the UAV has more non-sensing time to adjust its trajectory for communication performance improvement. Moreover, the achievable rate of our proposed scheme under high sensing frequency achieves significant improvement as compared to that of the "Time division" scheme, since more sensing time can be fully utilized for data transmission. In particular, the performance of our proposed method is very close to the upper bound of (P1), which is obtained by solving (SP2) based on Lemma 2, thereby illustrating the near-optimality of the proposed solution. Besides, when the sensing frequency increases, the achievable rate with a higher beam pattern gain threshold ${\tilde{\Gamma}}$ degrades faster as compared to that with a lower threshold. The main reason is that a higher beam pattern gain threshold forces the UAV to perform sensing tasks at a location closer to the target, thereby resulting in increasing path loss within pure communication duration. Similarly, it can be seen from Fig.~\ref{figure3b} that as the beam pattern gain threshold ${\tilde{\Gamma}}$ increases, the achievable rate under higher sensing frequency degrades faster as compared to that with a higher one. In Fig.~\ref{figure3c}, it is shown that for our proposed scheme, the achievable rates under different ${\tilde{\Gamma}}$ are almost equal when the distance $D$ is less than 100 m, i.e., a better communication performance can be achieved when sensing the target closer to the user (receiver). In Fig.~{\ref{figure4}}, the achievable rate gain achieved by our proposed scheme over the "optimal precoder only" scheme increases as the maximum flight speed $V_{\max}$ increases, since the UAV could obtain better channel gain within a shorter flying time.
\begin{figure*}[htbp]
	\centering
	\vspace{-3mm}
	\subfigure[Achievable rate versus sensing frequency.]{
		\label{figure3a}
		\vspace{-3mm}
		\includegraphics[width=5.2cm]{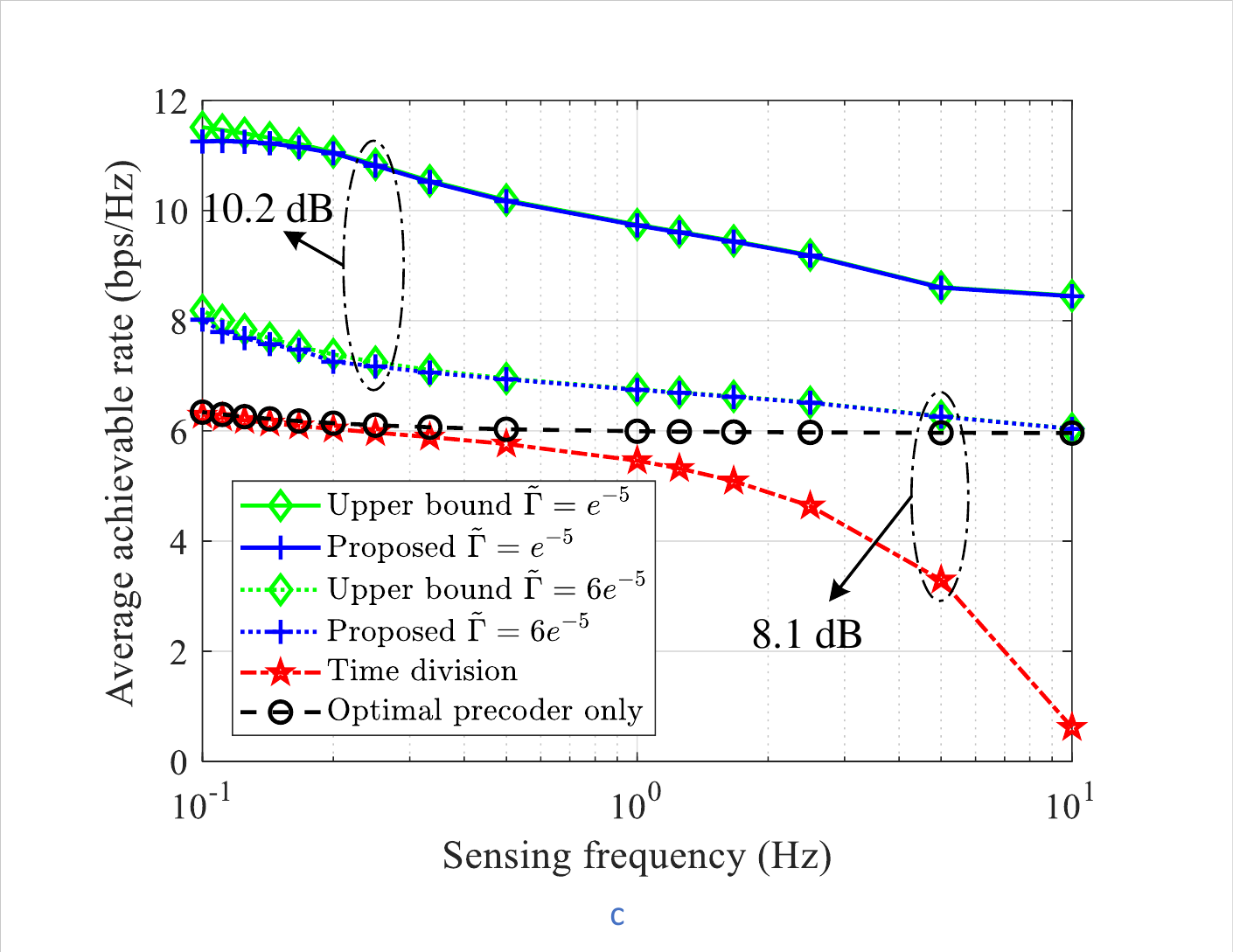}
	}
	\subfigure[Achievable rate versus beam pattern gain.]{
		\label{figure3b}
		\includegraphics[width=5.2cm]{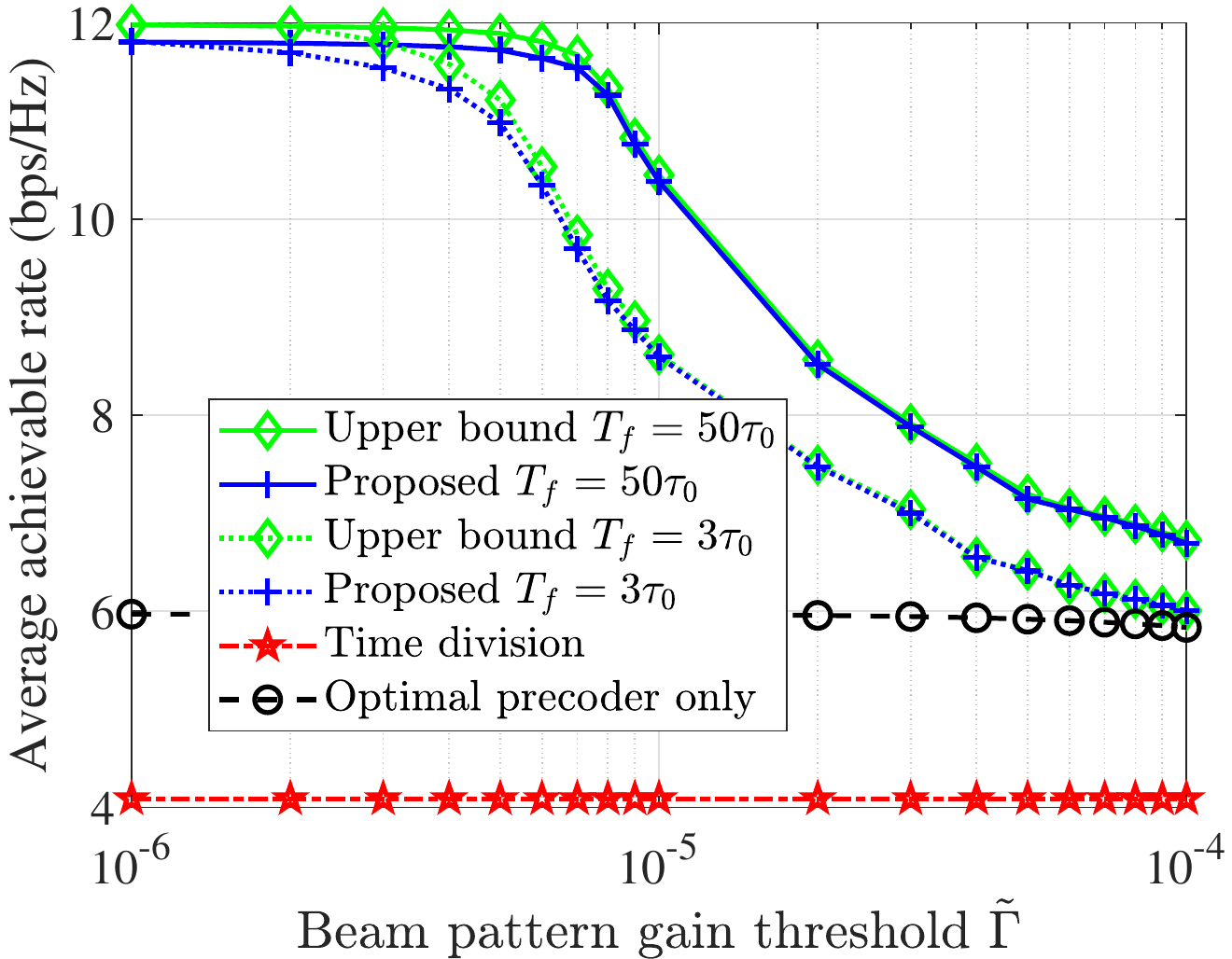}
	}
	\subfigure[Achievable rate versus distance.]{
		\label{figure3c}
		\includegraphics[width=5.2cm]{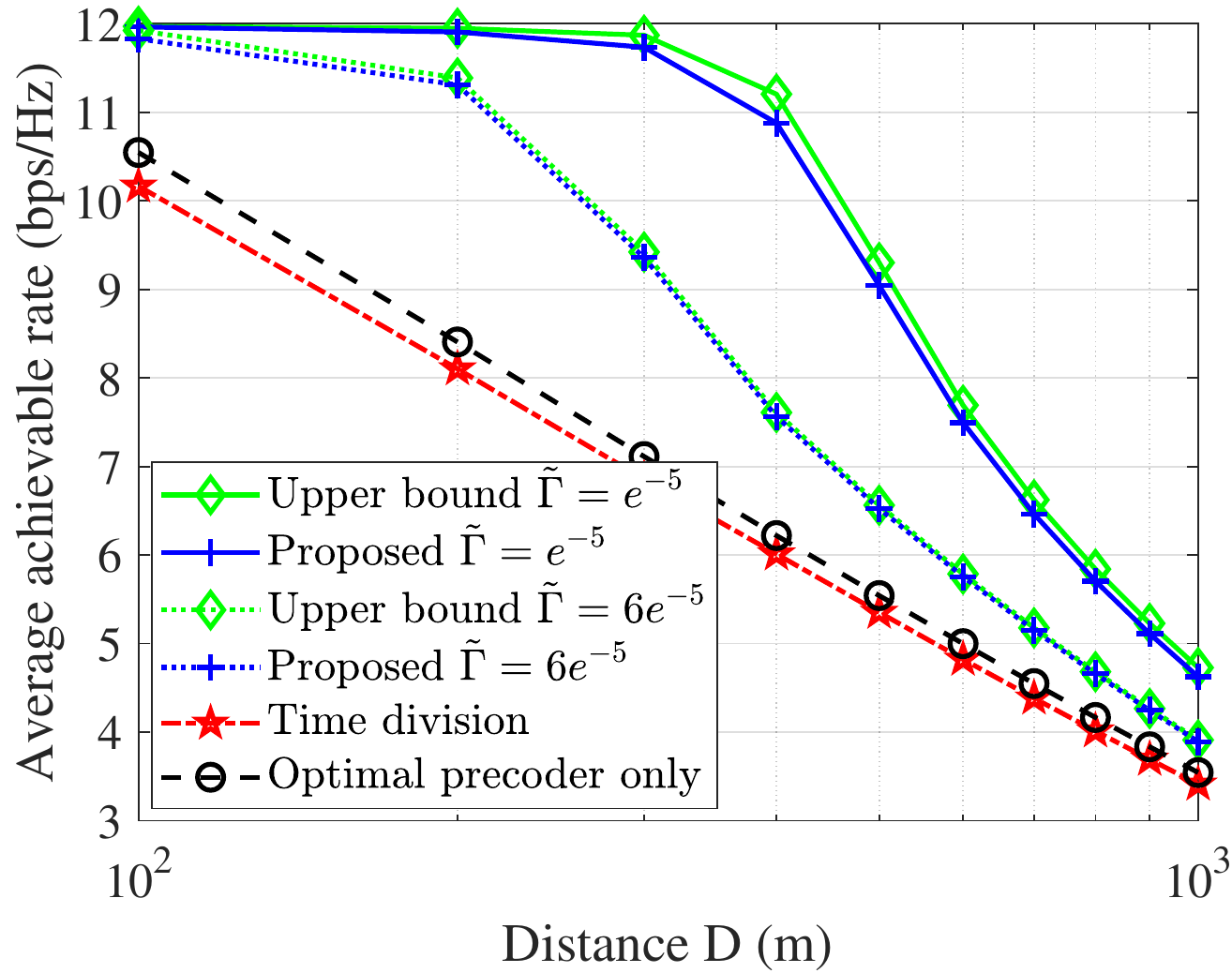}
	}
	\vspace{-3mm}
	\caption {Performance comparison among sensing power requirement, sensing frequency, and achievable rate (sensing frequency is defined as ${1 \mathord{\left/{\vphantom {1 {{T_f}}}} \right.\kern-\nulldelimiterspace} {{T_f}}}$).}
	\label{figure6}
	\vspace{-5mm}
\end{figure*}
\begin{figure}[t]
	\centering
	\setlength{\abovecaptionskip}{0.cm}
	\includegraphics[width=5cm]{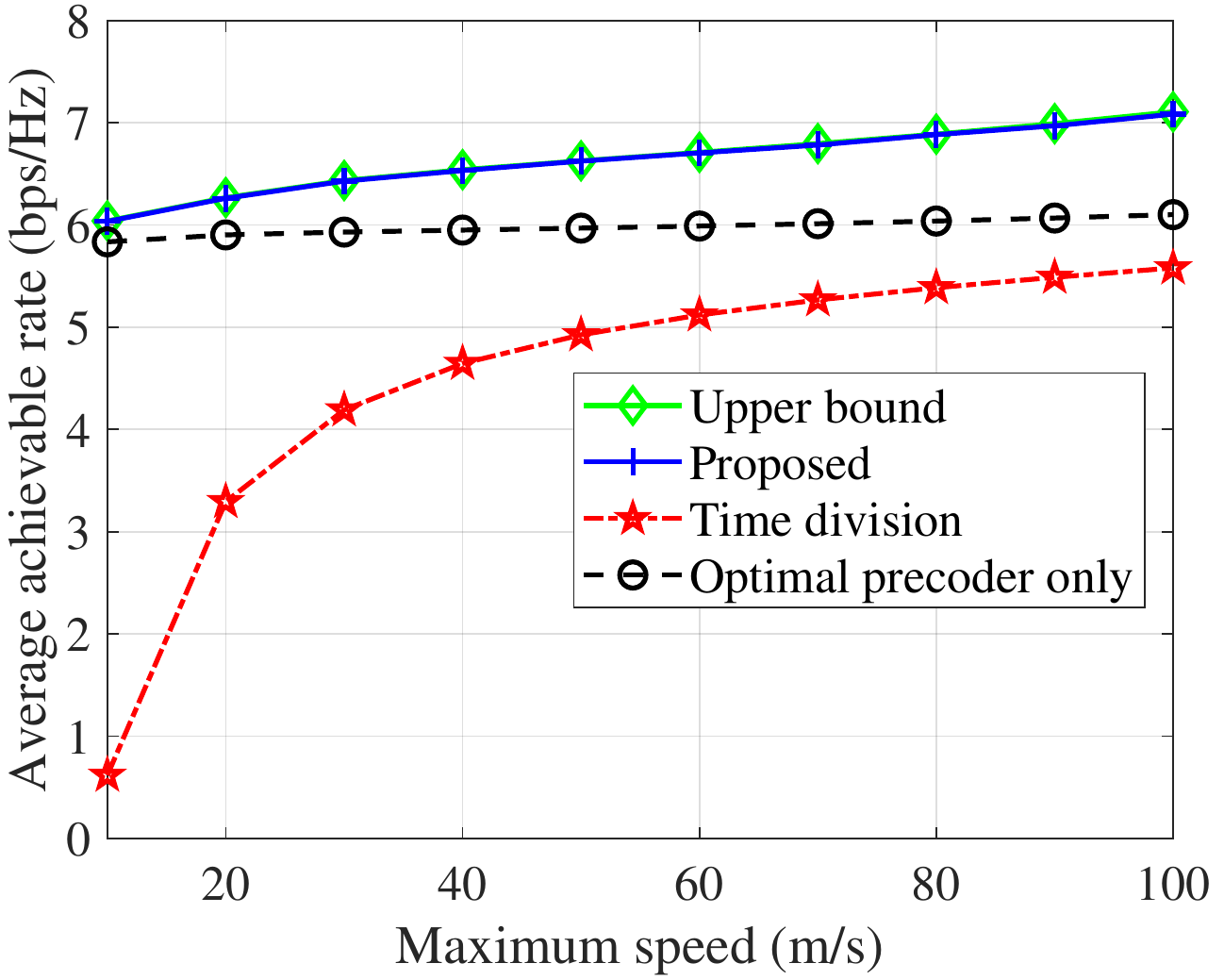}
	\vspace{-2mm}
	\caption{Achievable rate versus maximum flight speed.}
	\label{figure4}
	\vspace{-6mm}
\end{figure}

\vspace{-1.5mm}
\section{Conclusion and Future Works}
\vspace{-1mm}
This letter studied the achievable rate maximization problem in a UAV-enabled ISAC system. In particular, it was shown that there exists a novel structure-symmetry characteristic between optimal solutions in different ISAC frames without location constraints. Furthermore, a near-optimal location-constrained solution was achieved based on the derived closed-form achievable rate. Numerical results verified that the proposed scheme is able to drastically enlarge sensing and communication performance trade-off of UAV-enable ISAC systems. The more general 3D UAV trajectory optimization problems for multi-UAV ISAC scenarios are worthwhile future works.

\normalsize
\vspace{-1mm}
\section*{Appendix A:  \textsc{Proof of Proposition 1}}
\label{ProveP2}
\vspace{-1mm}
First, it can be easily shown that constraint (\ref{P2.1}b) is met with equality for the optimal solution since otherwise $\|{\bm{w}}_c\|$ can be always increased to improve the objective value until (\ref{P2.1}b) becomes active. Hence, constraint (\ref{P2.1}b) can be rewritten as ${\rm{tr}}\left( {\bm{W}} \right) = P_{\max}$. For $\frac{{{M{P_{\max }}\rho^2}}}{{{\left( {D - x} \right)}^2} + {H^2}} \ge  {\tilde{\Gamma}}$, We can readily derive that the beam pattern gain at target will be no less than the threshold ${\tilde{\Gamma}}$ if the optimal transmit precoder is $\sqrt{P_{\max}}\frac{{\bm{h}}_c}{\left\|{\bm{h}}_c\right\|}$, and its corresponding user SINR is $\frac{\beta_0{P_{\max}}M}{{{{x^2} + {H^2}}}}$. For $ \frac{{{M{P_{\max }}\rho^2}}}{{{\left( {D - x} \right)}^2} + {H^2}} <  {\tilde{\Gamma}}$, the Lagrangian function of (SDR2.1) is
\vspace{-1mm}
\begin{equation}
	\begin{aligned}
	L({\bm{W}},{\lambda _1},{\lambda _2} ) =&   -{\rm{tr}}( {{{\bm{H}}_c}{\bm{W}}} )  + {\lambda _1}( {{\rm{tr}}( {\bm{W}} ) - {P_{\max }}} ) \\
	& + {\lambda _2}( {{\tilde{\Gamma}}} - {\rm{tr}}({{{\bm{H}}_r}{\bm{W}}}) ).
	\end{aligned}
\vspace{-1mm}
\end{equation}
According to the Karush-Kuhn-Tucker (KKT) conditions, the optimal solution should satisfy 
\vspace{-1mm}
\begin{equation}\label{KKTcondition}
	 \left({\bm{H}}_c - {\lambda _1}{\bm{I}} + {\lambda _2}{\bm{H}}_r\right){\bm{W}} = 0,
	 \vspace{-1mm}
\end{equation}
which can be rewritten as
\vspace{-2mm}
\begin{equation}\label{TransformEquation}
	\left[ {\begin{array}{*{20}{c}}
			{{\bm{h}}_c}&{{{\bm{h}}_r}}
	\end{array}} \right]\left[ {\begin{array}{*{20}{c}}
			{{\bm{h}}_c^H{\bm{W}}}\\
			{{\lambda _2}{\bm{h}}_r^H{\bm{W}}}
	\end{array}} \right] = {\lambda _1}{\bm{W}}.
\vspace{-2mm}
\end{equation}
$\lambda_2$ is nonzero only when ${\rm{tr}}({{{\bm{H}}_r}{\bm{W}}}) < \tilde{\Gamma}$. In this case, MRT is the optimal solution to (P2). Define ${\bm{H}} = [{\bm{h}}_c, {\bm{h}}_r]$. By multiplying both sides of equation (\ref{TransformEquation}) with $\left[ {\begin{array}{*{20}{c}}
		{{{\bm{h}}_c}}&{\bm{0}}^{M\times1}\\
		{\bm{0}}^{M\times1}&{{{\bm{h}}_r}}
\end{array}} \right]({{\bm{H}}^H{\bm{H}}})^{-1}{\bm{H}}^H$ and taking the trace of both sides, it follows that
\vspace{-2mm}
\begin{equation}\label{EquationTrace}
 {\begin{array}{*{20}{c}}
			{{\rm{tr}}( {{\bm{H}}_c^H{\bm{W}}} ) \!=\! {\lambda _1}\frac{{{{| {{\bm{h}}_r^H} |}^2}{\rm{tr}}( {{{\bm{H}}_c}{\bm{W}}} ) - {\bm{h}}_c^H{{\bm{h}}_r}{\rm{tr}}( {{{\bm{h}}_c}{\bm{h}}_r^H{\bm{W}}} )}}{{{{| {{{\bm{h}}_c}} |}^2}{{| {{{\bm{h}}_r}} |}^2} - {{| {{\bm{h}}_c^H{{\bm{h}}_r}} |}^2}}}}\\
			{{\lambda _2}{\rm{tr}}( {{\bm{H}}_r^H{\bm{W}}} ) \!=\! {\lambda _1}\frac{{{{| {{\bm{h}}_c^H} |}^2}{\rm{tr}}( {{{\bm{H}}_r}{\bm{W}}} ) - {\bm{h}}_r^H{{\bm{h}}_c}{\rm{tr}}( {{{\bm{h}}_r}{\bm{h}}_c^H{\bm{W}}} )}}{{{{| {{{\bm{h}}_c}} |}^2}{{| {{{\bm{h}}_r}} |}^2} - {{| {{\bm{h}}_c^H{{\bm{h}}_r}} |}^2}}}}
	\end{array}} .
\vspace{-1mm}
\end{equation}
Since $\lambda_1$ and $\lambda_2$ are real-valued, ${\bm{h}}_c^H{{\bm{h}}_r}{\rm{tr}}( {{{\bm{h}}_c}{\bm{h}}_r^H{\bm{W}}} )$ and ${\bm{h}}_r^H{{\bm{h}}_c}{\rm{tr}}( {{{\bm{h}}_r}{\bm{h}}_c^H{\bm{W}}} )$ equal to $ {| {{\bm{h}}_c^H{{\bm{h}}_r}} |\sqrt {{\rm{tr}}( {{\bm{H}}_c^H{\bm{W}}} ){\rm{tr}}( {{{\bm{H}}_r^H}{\bm{W}}} )} }$ or $ -{| {{\bm{h}}_c^H{{\bm{h}}_r}} |\sqrt {{\rm{tr}}( {{\bm{H}}_c^H{\bm{W}}} ){\rm{tr}}( {{{\bm{H}}_r^H}{\bm{W}}} )} }$. Accordingly, by plugging equations in (\ref{EquationTrace}) into (\ref{KKTcondition}), it is readily derived that for any given nonzero $\lambda_1$ and $\lambda_2$, the optimal SNR at user ${\rm{tr}}\left( {{\bm{H}}_c^H{\bm{W}}^*} \right) = 	\frac{\gamma_0{({{( {D \!-\! x} )}^2} \!+\! {H^2}){\left( {\rho  {\sqrt {{{\tilde{\Gamma}}}}} {{ \!+\! }}\sqrt { {1 \!-\! {\rho ^2}} } {\sqrt {\frac{{M{P_{\max }}}}{{{{( {D \!-\! x} )}^2} \!+\! {H^2}}}\!-\! {{\tilde{\Gamma}}}}}   } \right)^2}}}{{{x^2} + {H^2}}}$, where $\rho  = \frac{| {\bm{a}}^H(x, {\bm{u}} ) {\bm{a}}(x, {\bm{v}} ) |}{\| {\bm{a}}^H(x, {\bm{u}} ) \|\| {\bm{a}}(x, {\bm{v}} )\| }$. Combining the above results yields the proof.
\footnotesize  	
\vspace{-1mm}
\bibliography{mybibfile}
\bibliographystyle{IEEEtran}

\end{document}